\RequirePackage{lineno}
\documentclass[aip,superscriptaddress,showpacs]{revtex4}

\usepackage[T1]{fontenc}
\usepackage{color,graphicx}
\usepackage{times}
\usepackage[colorlinks=false]{hyperref}
\usepackage{amsmath,amsthm,amssymb}
\usepackage{setspace}
\usepackage{cleveref}
\usepackage{palatino}
\usepackage{mathdots}
\usepackage{natbib}
\usepackage{doi}
\usepackage{diagbox}
\usepackage{mathtools}
\usepackage{stmaryrd}
\usepackage{booktabs}
\usepackage{multirow}
\usepackage{siunitx}
\usepackage{cprotect}
\usepackage{listings}
\usepackage[hang,small,bf]{caption}    % fancy captions
\usepackage{tikz}	
\usetikzlibrary{backgrounds,fit,decorations.pathreplacing}  % TikZ libraries
\usepackage[qm]{qcircuit}
\usepackage{scrextend}
\usepackage{hyperref}

\newcommand\ket[1]{\ensuremath{|#1\rangle}}
\newcommand\bra[1]{\ensuremath{\langle#1|}}

\newtheorem{definition}{Definition}
\newtheorem{lemma}{Lemma}
\newtheorem{theorem}{Theorem}
\newtheorem{corollary}{Corollary}

\definecolor{mGreen}{rgb}{0,0.6,0}
\definecolor{mGray}{rgb}{0.5,0.5,0.5}
\definecolor{mPurple}{rgb}{0.58,0,0.82}
\definecolor{backgroundColour}{rgb}{0.95,0.95,0.92}
\lstdefinestyle{CStyle}{
    backgroundcolor=\color{backgroundColour},   
    commentstyle=\color{mGreen},
    keywordstyle=\color{magenta},
    numberstyle=\tiny\color{mGray},
    stringstyle=\color{mPurple},
    basicstyle=\footnotesize,
    breakatwhitespace=false,         
    breaklines=true,                 
    captionpos=b,                    
    keepspaces=true,                 
    numbers=left,                    
    numbersep=5pt,                  
    showspaces=false,                
    showstringspaces=false,
    showtabs=false,                  
    tabsize=2,
    language=C
}

\begin{document}

\ifdefined\lineno
	\linenumbers
\else
\fi

%% End-Of-Header

\title{Optimizing T gates in Clifford+T circuit as $\pi/4$ rotations around Paulis}

\author{Fang Zhang}%
\email{fangzh@umich.edu}
\author{Jianxin Chen}%
\email{chenjianxin@tsinghua.org.cn}
\affiliation{Alibaba Quantum Laboratory, Alibaba Group, Bellevue, WA 98004, USA}%

\begin{abstract}
In this work, we introduce a new circuit optimization technique to reduce the number of \verb+T+ gates in Clifford+T circuits by treating \verb+T+ gates conjugated by Clifford gates as $\frac{\pi}{4}$-rotations around Pauli operators. The tested benchmarks shows up to $71.43\%$ and an average $42.67\%$ reduction in T-count, both surpass the best performance reported. The worst case complexity of our algorithm is $O(nk^2)$ where $n$ is the number of qubits and $k$ is the number of \verb+T+ gates in the original Clifford+\verb+T+ circuit.
% TODO
\end{abstract}

\date{\today}


\maketitle

\section{Introduction}

Much effort has been devoted to build the world's first meaningful quantum computer, which can deliver the ability to help scientists/engineers to develop new materials from drugs to battery, and to solve optimization problems from finance to logistics. In the past several years, prototypes of early-stage quantum computers are demonstrated by universities and tech companies. To understand the ability of these \emph{noisy intermediate-scale quantum} (NISQ) devices and the capability of future practical quantum computers, it is essential to develop efficient implementations of quantum algorithms, which are typically expressed in terms of quantum circuits. Quantum compilation, or more specifically, circuit optimization plays an important role in both regimes. 

There are many factors to take into consideration when executing quantum circuits on NISQ devices. Just like in classical computer science, efficiency measures the execution time necessary for an algorithm to finish its job.  Limited connectivity of existing quantum hardwares have a significant impact on the cost of quantum algorithms. This raise a related problem called \emph{qubit allocation} or \emph{quantum scheduling}\cite{siraichi2018qubit,li2018tackling}.  Furthermore, there may be significant variation in the error rate of qubits and links, or system reliability in general, on existing hardwares\cite{tannu2018case}. Such variation makes implementing a given circuit on NISQ hardwares to achieve better performance even more complicated. 

Among all the above-mentioned factors, execution time of a circuit is crucial since it will not only affect the performance but also decide if a circuit is realizable on hardware. Quantum states are very fragile due to decoherence which will limit the size of quantum circuits that can be executed on hardware. As a demonstration, various circuit optimization passes are implemented and tested in a recent paper\cite{CBS+18}, a measurable performance change of NISQ devices has been observed.

Building NISQ devices will not be our eventual goal. We expect to be able to protect quantum systems from various types of noises and scale up quantum devices using quantum error correction techniques. In the regime of fault tolerant quantum computing, one promising approach is to decompose logic operations on a surface code into Clifford+\verb+T+ circuits\cite{FMM+12}.  If we take the number of physical qubits involved (space cost) and protocol implementation time (time cost) into consideration, then \verb+T+ gate will be much more expensive than Clifford gates. For example, in a distance-$d$ surface code, the space-time cost of a logical \verb+T+ gate is of order $O(d^3)$ when employing Bravyi-Haah codes\cite{CTV17, FDJ13, OC17}. Although the cost of Clifford gates is of same order, the constant pre-factor is less by several orders of magnitude.

Various Clifford+\verb+T+ circuit optimization schemes have been proposed in literature\cite{gosset2013algorithm, amy2014polynomial, amy2016t, heyfron2018efficient}.  \cite{gosset2013algorithm} computes the \verb+T+-count exactly by meet-in-the-middle brute force search, which guarantees optimality, but does not go very far. Namely, it works for two qubits with $m=12$, or three qubits with $m=6$. It also introduces the notation of $R(\pm P)$, which is convenient for our approach.  $T$par \cite{amy2014polynomial} primarily optimizes \verb+T+-depth by resynthesizing the \verb+T+ gates in each \verb+CNOT++\verb+T+ section, but it also improves \verb+T+-count by cancelling pairs of \verb+T+ gates that are moved to the same place. It deals with Hadamard gate \verb+H+ by synthesizing all \verb+T+ gates that ``will not be computable'' before synthesizing the $H$, which is not quite optimal because some of those \verb+T+ gates may become ``computable'' again after the next \verb+H+ gate.  RM \cite{amy2016t} further optimizes \verb+CNOT++\verb+T+ sections by showing that optimizing \verb+T+-count in such circuits is equivalent to the minimum distance decoding problem of a certain Reed-Muller code. Although the decoding problem itself is hard, by using existing Reed-Muller decoders, RM is able to achieve an improved \verb+T+-count on many benchmark circuits. The handling of $H$ gates is still the same as $T$par.  TOpt \cite{heyfron2018efficient} uses the idea of gadgetization to deal with $H$ gates: Hadamard gates are eliminated from the first part of the circuit using an ancilla beginning in the $\ket{+}$ state, which is later measured and classically controls a Clifford gate (which may not be a Pauli, so the quantum controlled version may not be a Clifford). It also uses a new heuristic method on the resulting \verb+CNOT++\verb+T+ circuit to reduce \verb+T+-count in polynomial time.

In this paper, we adapt the $R(\pm P)$ notation used in \cite{gosset2013algorithm} to regard any Clifford+\verb+T+ circuit as a series of $\pi/4$ rotations, then cancel pairs of  \verb+T+ gates with the help of certain commutation relations. Compared to $T$par \cite{amy2014polynomial}, our method handles all Clifford gates in a unified way instead of having to consider \verb+H+ as a special case. This allows our method to achieve a strictly higher \verb+T+-count reduction.

The tested benchmarks shows up to $71.43\%$ and an average $42.67\%$ reduction in \verb+T+-count, both surpassing the best performance reported. It is worth mentioning that unlike other T-count optimizers which usually cause more than 100\% increase in \verb+CNOT+-count, our optimization procedure will not increase \verb+CNOT+-count while reducing \verb+T+-count. Furthermore, one may apply our optimization procedure to those optimized circuits generated from other circuit optimization tools no matter it optimizes \verb+CNOT+-count or \verb+T+-count, to get an even smaller \verb+T+-count without affecting other performance parameters (\verb+CNOT+-count). 

The detailed benchmarking result can be found in Sec \ref{sec:benchmark}.

% TODO

\section{Preliminaries}

In this section we establish some necessary notations needed for presenting our work. Throughout this paper, we use \verb+U+ to denote a unitary gate and use $U$ to denote the matrix corresponds to \verb+U+.

The Pauli operators on a single qubit are $I=\left(\begin{array}{cc} 1 & 0 \\ 0 & 1 \end{array}\right )$, $X=\left(\begin{array}{cc} 0 & 1 \\ 1 & 0 \end{array}\right )$, $Y=\left(\begin{array}{cc} 0 & -i \\ i & 0 \end{array}\right )$, $Z=\left(\begin{array}{cc} 1 & 0 \\ 0 & -1 \end{array}\right )$. The set of Pauli operators on $n$ qubits is defined as $\mathcal{P}_n=\{ \sigma_1\otimes \cdots \otimes \sigma_n| \sigma_i \in \{I, X, Y, Z\}\}$, and $\vert \mathcal{P}_n\vert=4^n$. We use subscripts to indicate qubits on which an operator acts. For example, in a $3$-qubit system, $X_1Z_2=X_1\times Z_2=(X \otimes I \otimes I)\times (I \otimes Z \otimes I)= X \otimes Z \otimes I$.

The single-qubit Clifford group is generated by the Hadamard gate \verb+H+ and phase gate \verb+S+ where $H=\frac{1}{\sqrt{2}}\left(\begin{array}{cc} 1 & 1 \\ 1 & -1 \end{array}\right )$, $S=\left(\begin{array}{cc} 1 & 0 \\ 0 & i \end{array}\right )$. Multi-qubit Clifford group is generated by these two gates along with the \verb+CNOT+ gate ($\ket{0}\bra{0}\otimes I+\ket{1}\bra{1}\otimes X$), acting on any, $1$ or $2$, of those $n$ qubits.

It is worth mentioning that the Clifford+\verb+T+ universal set attracts a lot of attention and has been widely studied in literature\cite{NRS2001, NRS+2006}. ``Universal'' here means any multi-qubit unitary gate can be implemented to any given accuracy by using a sequence of gates from Clifford+\verb+T+ set. The number of \verb+T+ gates appear in this sequence is defined as \verb+T+-count of that Clifford+\verb+T+ circuit representation.  Similarly, we can define \verb+CNOT+-count.

Some \verb+T+ gates on distinct qubits may be performed simultaneously; we will say these \verb+T+ gates are in the same \verb+T+-cycle. For a given Clifford+\verb+T+ representation, we may have different ways to group these \verb+T+ gates into \verb+T+-cycles. The minimum number of \verb+T+-cycles is defined as \verb+T+-depth of that representation.

A fundamental problem arise: Given a quantum algorithm represented as a unitary gate, what is the most hardware-efficient way to implement it?  As we explained in the introduction, much effort have been devoted to understand the hardware efficiency. The choice of cost function may depend on the architecture of the hardware.

Let's assume a unitary \verb+U+ has an exact representation over the Clifford+\verb+T+ gate set. Thus there exist Clifford operators $C_0‘, C_1’, \cdots, C_k‘$ such that
\begin{eqnarray}
U&=&C_0'(T^{\otimes m_1})C_1'(T^{\otimes m_2}) \cdots C_{k-1}' (T^{\otimes m_k}) C_k' \nonumber \\
&=& C_0' (T^{\otimes m_1}){C_0'}^{\dagger} (C_0'C_1') (T^{\otimes m_2}) (C_0'C_1')^{\dagger} \cdots (C_0'\cdots C_{k-1}')(T^{\otimes m_k})(C_0'\cdots C_{k-1}')^{\dagger} (C_0'\cdots C_{k-1}'C_k').
\end{eqnarray}

By substituting variables, there exist Clifford operators $C_0=C_0', C_1=C_0'C_1', \cdots, C_k=C_0'C_1'\cdots C_k'$ such that
\begin{eqnarray}
U&=& C_0 (T^{\otimes m_1}){C_0}^{\dagger} C_1 (T^{\otimes m_2}) (C_1)^{\dagger} \cdots (C_{k-1})(T^{\otimes m_k})(C_{k-1})^{\dagger} (C_k).
\end{eqnarray}

Each $C_i (T_1\otimes T_2\otimes \cdots T_{m_{i+1}}){C_i}^{\dagger}$ can be further written as $C_iT_1C_i^{\dagger} C_iT_2 C_i^{\dagger} \cdots C_iT_{m_{i+1}}C_i^{\dagger}$, a sequence of \verb+T+ gates acting on different qubits conjugated by Clifford operator ${\verb+C+}_i$.

\cprotect\section{\verb+T+ gates as $\pi/4$ rotations}

The \verb+T+ gate can be defined as a $\pi/4$ phase rotation around the single-qubit Pauli \verb+Z+:
\begin{equation}
\begin{split}
T=\ket{0}\bra{0}+e^{i\pi/4}\ket{1}\bra{1}=\frac{1+e^{i\pi/4}}{2}I+\frac{1-e^{i\pi/4}}{2}Z.
\end{split}
\end{equation}

The key insight is that, when the \verb+T+ gate acting on the $i$-th qubit is conjugated by any multi-qubit Clifford \verb+C+, it remains a $\pi/4$ phase rotation around a multi-qubit Pauli:
\begin{equation}
CT_iC^\dagger = \frac{1+e^{i\pi/4}}{2}I+\frac{1-e^{i\pi/4}}{2}CZ_iC^\dagger = R(CZ_iC^\dagger)
\end{equation}
where we borrow the following notation from \cite{gosset2013algorithm}:
\begin{equation}
R(P) = \frac{1+e^{i\pi/4}}{2}I+\frac{1-e^{i\pi/4}}{2}P,\qquad P\in\pm\mathcal{P}_n.
\end{equation}
By the definition of the Clifford group, $CZ_iC^\dagger\in\pm\mathcal{P}_n^*$, where $\mathcal{P}_n^* = \mathcal{P}_n\backslash\{I\}$. Also notice that
\begin{equation}
R(I) = I,\qquad R(-I) = e^{i\pi/4}I.
\end{equation}
While those are not $\pi/4$ rotations, they are both the identity up to a phase, and we choose to allow them because they may be convenient in some circuit transformations. % Maybe not; we can remove this mention depending on what ends up in the paper.

Another way of understanding $R(P)$ is that the set $\{R(P)\}$ is invariant under commutation with any Clifford \verb+C+:
\begin{equation}
\label{eq:commutation_with_clifford}
CR(P) = CR(P)C^\dagger C = R(CPC^\dagger)C.
\end{equation}
By applying \eqref{eq:commutation_with_clifford} to shift all Clifford gates to the back, any Clifford+\verb+T+ circuit can be represented as a series of $\pi/4$ rotations followed by a Clifford (up to a global phase):
\begin{equation}
\label{eq:product_of_rotations}
U = e^{i\phi}\left(\prod_{j=1}^m{R(P_j)}\right)C_0
\end{equation}
where $P_j\in\pm\mathcal{P}_n^*$, $C_0\in\mathcal{C}_n$, and $m$ is the number of \verb+T+ gates in the circuit. This same form is also presented in \cite{gosset2013algorithm}, except that here we allow negative Paulis as $P_j$ (because they are sometimes convenient), and we do not assume that $m$ is the \emph{minimum} \verb+T+-count (because computing the minimum is likely to be a computationally hard problem).

\section{A simple algorithm for optimizing the $T$-count}
\label{sec:t_count}
Once we have written a $\text{Clifford}+T$ circuit in the form of \eqref{eq:product_of_rotations}, there are several transformations we can apply to simplify it. The most obvious ones are:
\begin{equation}
\label{eq:cancellation}
R(P)R(-P) = I,\qquad R(P)R(P) = R(P)^2 = \frac{1+i}{2}(1-iP)\in\mathcal{C}_n,
\end{equation}
i.e. two opposite $\pi/4$ rotations will cancel each other, and two identical $\pi/4$ rotations can be combined into a $\frac{\pi}{2}$ rotation, which is a Clifford (namely, it is $CS_1C^\dagger$, for any Clifford $C$ such that $CZ_1C^\dagger = P$) and can be shifted to the back with \eqref{eq:commutation_with_clifford}. Applying either transformation will reduce the \verb+T+-count of the circuit by 2.

Of course, it is rare for a $\text{Clifford}+T$ circuit to have two $\pi/4$ rotations around the same Pauli directly adjacent to each other. However, it is not rare for two adjacent $\pi/4$ rotations to commute with each other. For example, if a group of \verb+T+ gates can be executed in parallel, then their corresponding $\pi/4$ rotations all commute pairwise. In the $R(P)$ representation, it is easy to determine which rotations commute:
\begin{equation}
\label{eq:commutation}
[R(P), R(Q)] = 0 \iff [P, Q] = 0. % Seems easy enough to see, and the intermediate steps are slightly long because I didn't define the coefficients $a$ and $b$ as constant. I can expand it if necessary.
\end{equation}
As it turns out, when we take \eqref{eq:commutation} into account, there are many more opportunities to apply \eqref{eq:cancellation}. Specifically, in the \verb+CNOT++\verb+T+ circuits studied in \cite{amy2014polynomial}, all the ``rotation axis'' Paulis are either $I$ or $Z$ on every single qubit; hence, as \cite{amy2014polynomial} shows, they all commute with each other.

By studying the $\text{Clifford}+T$ circuit as a whole, however, we can find more commuting $\pi/4$ rotations than \cite{amy2014polynomial}. For example, in the $\text{Mod}\;5_4$ circuit, there are several appearances of $(I\otimes H)CNOT(I\otimes H)$, which is really a $CZ$ gate in disguise. It is possible to rewrite the $CZ$ gate as a $CNOT+S$ circuit: % TODO: This part should really be explained with circuit diagrams, instead of (or in addition to) equations. People like to see circuit diagrams whether or not it is actually clearer.
\begin{equation}
\label{eq:cz_as_cnot_s_circuit}
CZ = (S\otimes S)CNOT(I\otimes S^\dagger)CNOT.
\end{equation}
If we apply this transformation before feeding the circuit to the algorithm of \cite{amy2014polynomial}, then it correctly recognizes that there is a $CNOT+T$ section of the circuit encompassing all the \verb+T+ gates, and optimizes the \verb+T+-count down to 8, as opposed to 16 when this transformation is not applied.

Below, we describe our algorithm that applies \eqref{eq:cancellation} and \eqref{eq:commutation} to reduce the \verb+T+-count. We note that our algorithm does not explicitly use any transformation like \eqref{eq:cz_as_cnot_s_circuit}; instead, because it first transforms the circuit into the form of \eqref{eq:product_of_rotations}, it treats all equivalent Clifford circuits in the same way, ensuring that the representation of the CZ circuit does not affect the result.

\begin{enumerate}
\item Let $U_0$ be an empty circuit, and $U_1$ be the input circuit transformed into the form of \eqref{eq:product_of_rotations}, with the phase factor $e^{i\phi}$ removed.
\begin{itemize}
\item Over the course of the algorithm, both $U_0$ and $U_1$ will be kept in the form of \eqref{eq:product_of_rotations}:
\begin{equation}
U_0 = \left(\prod_{j} R(P_{0, j})\right) C_0,\qquad U_1 = \left(\prod_{j} R(P_{1, j})\right) C_1.
\end{equation}
\item $U_0U_1$ should be a loop invariant.
\end{itemize}
\item While the $R(P)$ product part of $U_1$ is not empty:
\begin{enumerate}
\item Remove the leftmost factor $R(P)$ from $U_1$, and insert $R(C_0PC_0^\dagger)$ to the left of $C_0$ in $U_0$.
\item Scan through each other factor $R(Q)$ in $U_0$, stopping when any of the following is true:
\begin{itemize}
\item $Q = \pm C_0PC_0^\dagger$.
\item $[Q, C_0PC_0^\dagger] \ne 0$. (Equivalently, $\{Q, C_0PC_0^\dagger\} = 0$.)
\item Every other factor $R(Q)$ in $U_0$ is exhausted.
\end{itemize}
\item If we found some $R(Q)$ such that $Q=\pm C_0PC_0^\dagger$:
\begin{itemize}
\item If $Q = -C_0PC_0^\dagger$, simply remove $R(Q)$ and $R(C_0PC_0^\dagger)$ from $U_0$.
\item If $Q = +C_0PC_0^\dagger$, then remove $R(Q)$ and $R(C_0PC_0^\dagger)$ from $U_0$, and let $C_0\gets R(Q)^2C_0$.
\end{itemize}
\end{enumerate}
\item Let $C_0 \gets C_0C_1$, and $C_1 \gets I$.
\item Transform $U_0$ back into a $\text{Clifford}+T$ circuit, and return it as the optimized circuit.
\end{enumerate}

% Probably we should describe the algorithm in an algorithm environment or as pseudocode?
The bottleneck of this algorithm is step~2(b); in the worst case, for each $R(P)$ factor in the input circuit, we need to scan through every $R(Q)$ factor in $U_0$, checking equality and commutativity each time. Therefore the worst case complexity is $O(k^2n)$, where $n$ is the number of qubits and $k$ is the number of \verb+T+ gates in the original circuit.

\section{The structure of $\pi/4$ rotations as a DAG and application to $T$-depth optimization}

As we have seen, the commutation relation \eqref{eq:commutation} proves quite useful in manipulating $\text{Clifford}+T$ circuits. Although it does not decrease the \verb+T+-count on its own, it enables other rules like \eqref{eq:cancellation} to be applied. % Other examples, like the 15 T rule?
Furthermore, reordering \verb+T+ gates may be interesting for its own sake if we consider other goals of circuit optimization, such as minimizing the \verb+T+-depth.

Since \eqref{eq:commutation} just swaps adjacent $\pi/4$ rotations, applying \eqref{eq:commutation} repeatedly is equivalent to applying some permutation to all the $\pi/4$ rotations in a circuit. However, it may not be immediately obvious which permutations are permissible. Fortunately, there is a simple description for the set of permissible permutations if we regard the structure of $\pi/4$ rotations as a directed acyclic graph (DAG).

\begin{definition}
The \verb+T+-graph of a circuit $U$ in the form of
\begin{equation}
\label{eq:product_of_rotations_2}
U = e^{i\phi}\left(\prod_{j=1}^m{R(P_j)}\right)C_0
\end{equation}
is defined as $G_T(U) = (V, E)$, where $V = \{v_j\}$ for $j=1, \cdots, m$, and
\begin{equation}
E = \{(v_i, v_j) \mid [P_i, P_j] \ne 0 \wedge i < j\},
\end{equation}
i.e. $G_T(U)$ has a vertex for each $\pi/4$ rotation in $U$, and an edge for each pair of $\pi/4$ rotations whose Pauli anti-commute with each other, with the direction of the edge determined on the order in which the two rotations appear in the original circuit.
\end{definition}

The \verb+T+-graph $G_T(U)$ of a circuit is always a DAG, since the order $v_1, \cdots, v_m$ is a topological ordering of $G_T(U)$. In fact, we have the following stronger result:

\begin{theorem}
\label{thm:topological_ordering}
By applying only \eqref{eq:commutation}, a circuit in the form of \eqref{eq:product_of_rotations_2} can be transformed into
\begin{equation}
\label{eq:product_of_rotations_permuted}
e^{i\phi}\left(\prod_{j=1}^m{R(P_{p_j})}\right)C_0
\end{equation}
If and only if $v_{p_1}, \cdots, v_{p_m}$ is a topological ordering of $G_T(U)$.
\end{theorem}

\begin{proof}
First, suppose that repeated applications of \eqref{eq:commutation} can indeed transform \eqref{eq:product_of_rotations_2} into \eqref{eq:product_of_rotations_permuted}. Let $e = (v_{p_i}, v_{p_j})$ be any one edge of $G_T(U)$. By the definition of $G_T(U)$, in the original circuit  \eqref{eq:product_of_rotations_2}:
\begin{itemize}
\item $R(P_{p_i})$ appears before $R(P_{p_j})$;
\item $[P_{p_i}, P_{p_j}] \ne 0$.
\end{itemize}
The latter condition means that application of \eqref{eq:commutation} cannot change the relative order of $R(P_{p_i})$ and $R(P_{p_j})$, so in \eqref{eq:product_of_rotations_permuted}, $R(P_{p_i})$ still appears before $R(P_{p_j})$, i.e. $i < j$. Since this argument is valid for every edge $e = (v_{p_i}, v_{p_j})$ of $G_T(U)$, it follows that $v_{p_1}, \cdots, v_{p_m}$ is a topological ordering of $G_T(U)$.

Conversely, suppose that $v_{p_1}, \cdots, v_{p_m}$ is indeed a topological ordering of $G_T(U)$. Then \eqref{eq:product_of_rotations_2} can be transformed into \eqref{eq:product_of_rotations_permuted} by doing a bubble sort on the $\pi/4$ rotations until they match the order $R(P_{p_1}), \cdots, R(P_{p_m})$. At each time step, such a bubble sort will only swap two adjacent rotations, $R(P_{p_j})$ and $R(P_{p_i})$, if both of the following are true:
\begin{itemize}
\item $p_j < p_i$ (so in \eqref{eq:product_of_rotations_2} $R(P_{p_j})$ comes before $R(P_{p_i})$);
\item $i < j$ (so in \eqref{eq:product_of_rotations_permuted} $R(P_{p_i})$ comes before $R(P_{p_j})$).
\end{itemize}
The first condition guarantees that $(v_{p_i}, v_{p_j})$ is not an edge of $G_T(U)$, and the second condition, together with the assumption that $v_{p_1}, \cdots, v_{p_m}$ is a topological ordering, guarantees that $(v_{p_j}, v_{p_i})$ is not an edge of $G_T(U)$ either. Therefore $[P_{p_i}, P_{p_j}] = 0$, and \eqref{eq:commutation} can indeed be applied to make the swap.
\end{proof}

We note that the ``only if'' direction of Theorem~\ref{thm:topological_ordering} allows \emph{only} applications of \eqref{eq:commutation}. For example, Both $R(Z)R(-Z)R(X)R(-X)$ and $R(Z)R(X)R(-X)R(-Z)$ evaluates to $I$ by \eqref{eq:cancellation}, so they are indeed equivalent, but that is not covered by Theorem~\ref{thm:topological_ordering}.

% TODO: At this point we can (more easily) prove the optimality of the T-count optimizing algorithm, under the condition of using only cancellation and commutation.
% TODO: Extracting a phase gate from changing $R(P)$ to $R(-P)$ does not change the transitive closure of $G_T(U)$.

We have alluded to the relation between the representation \eqref{eq:product_of_rotations} and \verb+T+-depth of a circuit when we mentioned that \verb+T+ gates that can be applied in parallel always translate to commuting $\pi/4$ rotations. Unfortunately, the strong converse is not always true --- a group of pairwise commuting $\pi/4$ rotations cannot always be translated into a group of \verb+T+ gates applied in parallel. Instead, we need an extra condition:

\begin{theorem}
\label{thm:layer}
A group of $\pi/4$ rotations, $\prod_{j=1}^m R(P_j)$, can be translated into a $\text{Clifford}+T$ circuit of \verb+T+-depth 1 (i.e. one where all the \verb+T+ gates can be applied in parallel) if both of the following are true:
\begin{itemize}
\item All $P_j$ commute with each other;
\item There does not exist a non-empty subset $S$ of $\{1, \cdots, m\}$ such that $\prod_{j\in S}P_j = \pm I$.
\end{itemize}
\end{theorem}

\begin{proof}
Those two conditions are exactly the conditions under which there exists a Clifford $C\in\mathcal{C}_n$ such that $CP_jC^\dagger = Z_j$ holds for $j = 1, \cdots, m$. % Do I need to explain, or cite something else?
The circuit is then given by
\begin{equation}
C^\dagger\left(T^{\otimes m}\otimes I^{\otimes(n-m)}\right)C = C^\dagger\left(\prod_{j=1}^m R(Z_j)\right)C = \prod_{j=1}^m R(C^\dagger Z_j C) = \prod_{j=1}^m R(P_j).
\end{equation}
\end{proof}

Notice that translating a \verb+T+-depth 1 circuit back into $\pi/4$ rotations with \eqref{eq:commutation_with_clifford} will indeed result in a circuit in the form of \eqref{eq:product_of_rotations} satisfying the conditions of Theorem~\ref{thm:layer}. However, as written, the converse of Theorem~\ref{thm:layer} is not true, because a circuit that initially does not satisfy those conditions may be simplified or otherwise transformed to satisfy them.

As \cite{amy2014polynomial} points out, when the goal is to minimize \verb+T+-depth, additional ancilla qubits initialized in the $\ket{0}$ state may be useful. However, they are also tricky: when using ancilla qubits to implement a unitary, they must return to a constant state at the end of the circuit, or otherwise the output state may be entangled with the ancilla instead of being a pure state. % Need to explain further?

A obvious way to satisfy this condition is to ensure that, when the circuit is transformed into the form of \eqref{eq:product_of_rotations}, all the ancilla qubits stay in the state $\ket{0}$ throughout the circuit. It turns out that this is more or less what \cite{amy2014polynomial} does, too.

\begin{lemma}
Let $P = P_0\otimes Q$, where $P \in \mathcal{P}_{n+t}$, $P_0 \in \mathcal{P}_n$, and $Q \in \mathcal{P}_t$. Then the following statements are equivalent:
\begin{enumerate}
\item $\forall\ket{\psi}\; \exists\ket{\varphi}\; R(P)(\ket{\psi}\otimes\ket{0}^{\otimes t}) = \ket{\varphi}\otimes\ket{0}^{\otimes t}$;
\item $\exists\ket{\psi}\; \exists\ket{\varphi}\; R(P)(\ket{\psi}\otimes\ket{0}^{\otimes t}) = \ket{\varphi}\otimes\ket{0}^{\otimes t}$;
\item $Q \in \{I, Z\}^{\otimes t}$;
\item $\forall\ket{\psi}\; R(P)(\ket{\psi}\otimes\ket{0}^{\otimes t}) = (R(P_0)\ket{\psi})\otimes\ket{0}^{\otimes t}$.
\end{enumerate}
\end{lemma}

\begin{proof}
$(1) \Rightarrow (2)$: Obvious.
\medskip

$(2) \Rightarrow (3)$: The left hand side can be expanded as follows:
\begin{equation}
\begin{split}
R(P)(\ket{\psi}\otimes\ket{0}^{\otimes t}) &= \left(\frac{1+e^{i\pi/4}}{2}I+\frac{1-e^{i\pi/4}}{2}P\right)(\ket{\psi}\otimes\ket{0}^{\otimes t})\\
&= \frac{1+e^{i\pi/4}}{2}(\ket{\psi}\otimes\ket{0}^{\otimes t})+\frac{1-e^{i\pi/4}}{2}P(\ket{\psi}\otimes\ket{0}^{\otimes t})\\
&= \frac{1+e^{i\pi/4}}{2}(\ket{\psi}\otimes\ket{0}^{\otimes t})+\frac{1-e^{i\pi/4}}{2}(P_0\ket{\psi}\otimes Q\ket{0}^{\otimes t}).
\end{split}
\end{equation}
The last line should be equal to $\ket{\varphi}\otimes\ket{0}^{\otimes t}$. Equivalently, both of the following must hold:
\begin{align}
\label{eq:ancilla_invariant}
Q\ket{0}^{\otimes t} &= \ket{0}^{\otimes t},\\
\label{eq:ancilla_equivalence}
\frac{1+e^{i\pi/4}}{2}\ket{\psi}+\frac{1-e^{i\pi/4}}{2}P_0\ket{\psi} &= \ket{\varphi}.
\end{align}
The desired statement, $Q \in \{I, Z\}^{\otimes t}$, follows from \eqref{eq:ancilla_invariant}.
\medskip

$(3) \Rightarrow (4)$: Notice that, in \eqref{eq:ancilla_equivalence}, the left hand side evaluates to $R(P_0)\ket{\psi}$. The rest easily follows.
\medskip

$(4) \Rightarrow (1)$: Obvious.
\end{proof}

By repeated application of ``$(3) \Rightarrow (4)$'', we immediately get an easy way to incorporate ancilla qubits into our circuits:

\begin{corollary}
\label{cor:ancilla}
A circuit in the form of \eqref{eq:product_of_rotations} implementing the unitary $U$ can always be extended with \verb+T+ ancilla qubits as follows:
\begin{equation}
\label{eq:product_of_rotations}
U' = e^{i\phi}\left(\prod_{j=1}^m{R(P_j\otimes Q_j)}\right)(C_0\otimes I^{\otimes t})
\end{equation}
where $Q_j \in \{I, Z\}^{\otimes t}$. The extension satisfies $U'(\ket{\psi}\otimes\ket{0}^{\otimes t}) = (U\ket{\psi})\otimes\ket{0}^{\otimes t}$ for all $\ket{\psi}$.
\end{corollary}

\begin{theorem}
\label{thm:ancilla_t_graph}
Adding ancilla qubits to a circuit $U$ using Corollary~\ref{cor:ancilla} does not change its \verb+T+-graph $G_T(U)$.
\end{theorem}

\begin{proof}
It suffices to show that $G_T(U)$ and $G_T(U')$ have the same set of edges, which is equivalent to showing that $[P_i\otimes Q_i, P_j\otimes Q_j] = 0$ if and only if $[P_i, P_j] = 0$. Since $Q_i, Q_j \in \{I, Z\}^{\otimes t}$, $[Q_i, Q_j] = 0$ always holds, so indeed we have $[P_i\otimes Q_i, P_j\otimes Q_j] = 0 \iff [P_i, P_j] = 0$.
\end{proof}

Now we revisit Theorem~\ref{thm:layer}. With the help of ancilla qubits, we can remove the second condition.

\begin{theorem}
\label{thm:layer_with_ancilla}
A group of $\pi/4$ rotations, $\prod_{j=1}^m R(P_j)$, can be translated into a $\text{Clifford}+T$ circuit of \verb+T+-depth 1, with $m$ ancilla qubits, as long as all $P_j$ commute with each other.
\end{theorem}

\begin{proof}
We apply Corollary~\ref{cor:ancilla} to extend the group of rotations into $\prod_{j=1}^m R(P_j\otimes Q_j)$, where $Q_j = Z_j \in \{I, Z\}^{\otimes m}$. Now the second condition of Theorem~\ref{thm:layer} is naturally satisfied, and the first condition of Theorem~\ref{thm:layer} will continue to be satisfied due to the proof of Theorem~\ref{thm:ancilla_t_graph}. Hence an application of Theorem~\ref{thm:layer} on the extended group of rotations give the desired result.
\end{proof}

For some $\text{Clifford}+T$ circuits, the combination of the commutation relation \eqref{eq:commutation} and Theorem~\ref{thm:layer_with_ancilla} is a powerful tool for optimizing the \verb+T+-depth (with an unlimited number of ancilla qubits).

\begin{theorem}
\label{thm:t_depth}
Using only \eqref{eq:commutation} and Theorem~\ref{thm:layer_with_ancilla}, the minimum \verb+T+-depth that can be achieved for a circuit $U$ is exactly equal to the length of (i.e. the number of vertices on) the longest path in $G_T(U)$.
\end{theorem}

\begin{proof}
First, we show that the minimum \verb+T+-depth cannot be less than the length of the longest path in $G_T(U)$. Suppose the longest path in $G_T(U)$ is $v_{j_1}, v_{j_2}, \cdots, v_{j_k}$. By the construction of $G_T(U)$, we have $\{P_{j_1}, P_{j_2}\} = \{P_{j_2}, P_{j_3}\} = \cdots = \{P_{j_{k-1}}, P_{j_k}\} = 0$. When using \eqref{eq:commutation} to swap the rotations, the relative order of $R(P_{j_1}), R(P_{j_2}), \cdots, R(P_{j_k})$ cannot change; then when using Theorem~\ref{thm:layer_with_ancilla} to group the rotations into layers, each of them must go into a different layer, since the rotations in the same layer must pairwise commute. Therefore the minimum \verb+T+-depth that can be achieved this way is $k$, the length of the longest path.

Next, we show that a \verb+T+-depth of $k$ is indeed achievable. To this end, we first sort the vertices of $G_T(U)$ in ascending order of the length of the longest path \emph{ending at} each vertex. Since the minimum length of such path is 1, and the maximum is $k$, this essentially divides the vertices of $G_T(U)$ into $k$ layers. Edges in $G_T(U)$ can only go from a layer to a \emph{later} layer, not to any prior layer nor to the same layer. Therefore, this sort is a topological ordering, and by Theorem~\ref{thm:topological_ordering}, we can reorder the $\pi/4$ rotations in $U$ accordingly. The fact that there does not exist any edge between vertices in the same layer also means that all rotations corresponding to those vertices commute with each other, so by applying Theorem~\ref{thm:layer_with_ancilla}, we can transform each layer into a \verb+T+-layer. The end result is a $\text{Clifford}+T$ circuit with \verb+T+-depth $k$.
\end{proof}

Again, Theorem~\ref{thm:t_depth} restricts the rules of transformation allowed. In particular, cancellation with \eqref{eq:cancellation} is not taken into account. In our experiments, we utilize cancellation by applying the algorithm described in \ref{sec:t_count} before applying Theorem~\ref{thm:t_depth}; this is also similar to the approach taken in \cite{amy2014polynomial}.

\section{Benchmarking}\label{sec:benchmark}

In this section, we will benchmark our T-Optimizer software with two leading Clifford+T optimizers, T-par and TOpt. The specifications of circuit inputs that T-par and TOpt take are slightly different. T-Optimizer takes a circuit in the .qc format that T-par accepted.

\begin{lstlisting}[basicstyle=\ttfamily, caption={An example : mod5\_4.qc}]
 .v b c d e a
 .i b c d e

 BEGIN
 X   a
 H   a
 Z   b   e   a
 Z   d   e   a
 H   a
 tof e   a
 H   a
 Z   c   d   a
 H   a
 tof d   a
 H   a
 Z   b   c   a
 H   a
 tof c   a
 tof b   a
 END
\end{lstlisting}

A description of the .qc format can be found at \href{https://github.com/aparent/QCViewer}{https://github.com/aparent/QCViewer}. For the sake of readability, we briefly explain the basics of .qc file, which indeed serves as textual representation of what the circuit consists of.  In the header, \verb+.v+ defines names of the circuit qubits, \verb+.i+ and \verb+.o+ specify which qubits accept primary inputs and report primary outputs respectively.  In the body,  \verb+H+ refers the Hadamard gate and \verb+tof+ refers to the Toffoli gate. The identifiers following the names refer to the qubits the gate is applies to. For instance, the line \verb+H a+ indicates a Hadamard gate on the qubit named \verb+a+. In the case where there is more than one qubit as with the Toffoli gate, the rightmost one is the target qubit and the rest are the control qubits. 

As mentioned in \cite{heyfron2018efficient}, \verb+T+-count does not account for the full space-time cost of quantum computation. Thus in our benchmarking result, we present both the \verb+CNOT+-count and the \verb+T+-count. The other two leading circuit optimization softwares, T-par and TOpt often produce the output circuit with increased \verb+CNOT+-count. We also aware that these two open source packages are under active development even after they reported their results in \cite{heyfron2018efficient, amy2016t, amy2014polynomial}. For some instances we tested, the output result is actually better than those reported number. Thus we choose to benchmark our T-Optimizer with these two open source packages rather than the reported results. TOpt has a heuristic subroutine, so we run each instance $100$ times, pick the best pair (\verb+CNOT+-count, \verb+T+-count) in the sense of reverse lexicographical order on the Cartesian product; if it is worse than the report number, then we will run $100$ more times.

\begin{figure}[h] 
\mbox{
\Qcircuit @C=1em @R=.7em {
& \gate{X} & \gate{H} 	& \gate{Z}	& \gate{Z}	& \gate{H}		&\targ	& \gate{H}	& \gate{Z}	& \gate{H}	& \targ	& \gate{H}	& \gate{Z}	& \gate{H}& \targ	& \targ 	&\qw	\\
& \qw	& \qw		& \ctrl{-1}	& \qw	& \qw		& \qw	& \qw	& \qw	& \qw	& \qw	& \qw	& \ctrl{-1}	& \qw	& \qw	& \ctrl{-1}	& \qw\\
& \qw 	& \qw		& \qw	& \qw	& \qw		& \qw	& \qw	& \ctrl{-2}	& \qw	& \qw	& \qw	& \ctrl{-1}	& \qw	& \ctrl{-2}	& \qw	& \qw\\
& \qw	& \qw		& \qw	& \ctrl{-3}	& \qw		& \qw	& \qw	& \ctrl{-1}	& \qw	& \ctrl{-3}	& \qw	& \qw	& \qw	& \qw	& \qw	& \qw\\
& \qw	& \qw		& \ctrl{-3}	& \ctrl{-1}	& \qw		& \ctrl{-4}	& \qw	& \qw	& \qw	& \qw 	& \qw	& \qw	& \qw	& \qw	& \qw	& \qw\\
}
} 
\caption{Quantum Circuit describe in mod5\_4.qc} 
\end{figure}
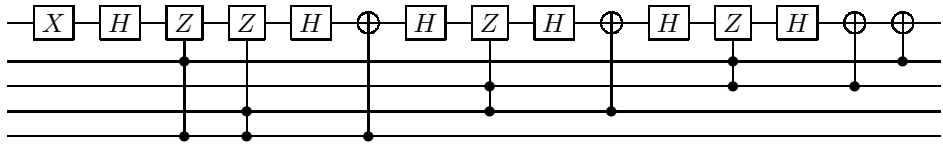

\begin{table}[t]
\scriptsize
\centering
\begin{tabular}{|l|c||r||r||r|}
\hline
Benchmark&(\verb+CNOT+-count, \verb+T+-count)&(\verb+CNOT+-count, \verb+T+-count)&(\verb+CNOT+-count, \verb+T+-count)&(\verb+CNOT+-count, \verb+T+-count) \\
&original&after T-par\cite{amy2014polynomial, amy2016t}&after TOpt\cite{heyfron2018efficient}& after T-Optimizer\\ \hline\hline
Mod 5$_{4}$ \cite{m11}					&(\textcolor{blue}{32},\textcolor{red}{28})		&(\textcolor{blue}{48},\textcolor{red}{16}) 					&(\textcolor{blue}{52},\textcolor{red}{18})	&(\textcolor{blue}{28},\textcolor{red}{8})		\\ \hline
VBE-Adder$_{3}$ \cite{vedral1996quantum}	&(\textcolor{blue}{80},\textcolor{red}{70})		&(\textcolor{blue}{114},\textcolor{red}{24})					&(\textcolor{blue}{106},\textcolor{red}{36})	&(\textcolor{blue}{70},\textcolor{red}{24})		\\ \hline
CSLA-MUX$_{3}$ \cite{van2005fast}			&(\textcolor{blue}{90},\textcolor{red}{70})		&(\textcolor{blue}{379},\textcolor{red}{62})					&(\textcolor{blue}{176},\textcolor{red}{64}) 	&(\textcolor{blue}{80},\textcolor{red}{62})		\\ \hline
CSUM-MUX$_{9}$ \cite{van2005fast}		&(\textcolor{blue}{196},\textcolor{red}{196})	&(\textcolor{blue}{366},\textcolor{red}{84}) 				&(\textcolor{blue}{984},\textcolor{red}{76})  	&(\textcolor{blue}{168},\textcolor{red}{84})		\\ \hline
QCLA-Com$_{7}$ \cite{DKR+06}			&(\textcolor{blue}{215},\textcolor{red}{203})	&(\textcolor{blue}{371},\textcolor{red}{94})					&(\textcolor{blue}{420},\textcolor{red}{135}) 	&(\textcolor{blue}{186},\textcolor{red}{95}) 	\\ \hline
QCLA-Mod$_{7}$ \cite{DKR+06}			&(\textcolor{blue}{441},\textcolor{red}{413})	&(\textcolor{blue}{813},\textcolor{red}{231})\footnote{\label{better}T-count better than reported.}			&(\textcolor{blue}{918},\textcolor{red}{305})  	&(\textcolor{blue}{382},\textcolor{red}{237})	\\ \hline
QCLA-Adder$_{10}$ \cite{DKR+06}			&(\textcolor{blue}{267},\textcolor{red}{238})	&(\textcolor{blue}{648},\textcolor{red}{162})				&(\textcolor{blue}{799},\textcolor{red}{181})\footref{better}  	&(\textcolor{blue}{233},\textcolor{red}{162})	\\ \hline
Adder$_{8}$ \cite{TTK10}					&(\textcolor{blue}{466},\textcolor{red}{399})	&(\textcolor{blue}{741},\textcolor{red}{213})				&(\textcolor{blue}{749},\textcolor{red}{280})\footref{better} 	&(\textcolor{blue}{409},\textcolor{red}{173})	\\ \hline
RC-Adder$_{6}$ \cite{CDK+04}				&(\textcolor{blue}{104},\textcolor{red}{77})		&(\textcolor{blue}{165},\textcolor{red}{47})					&(\textcolor{blue}{145},\textcolor{red}{59}) 	&(\textcolor{blue}{93},\textcolor{red}{47})		\\ \hline
Mod-Red$_{21}$ \cite{MS12}				&(\textcolor{blue}{122},\textcolor{red}{119})	&(\textcolor{blue}{223},\textcolor{red}{73}) 					&(\textcolor{blue}{168},\textcolor{red}{85}) 	&(\textcolor{blue}{105},\textcolor{red}{73})		\\ \hline
Mod-Mult$_{55}$ \cite{MS12}				&(\textcolor{blue}{55},\textcolor{red}{49})		&(\textcolor{blue}{104},\textcolor{red}{35}) 					&(\textcolor{blue}{104},\textcolor{red}{28}) 	&(\textcolor{blue}{48},\textcolor{red}{35})		\\ \hline
$\Lambda_3(X)$ -- \cite{BBC+95}			&(\textcolor{blue}{28},\textcolor{red}{28})		&(\textcolor{blue}{52},\textcolor{red}{16})					&(\textcolor{blue}{32},\textcolor{red}{22})	&(\textcolor{blue}{24},\textcolor{red}{16})		\\ \hline
\hspace{3.05em} -- \cite{NC04}				&(\textcolor{blue}{21},\textcolor{red}{21})		&(\textcolor{blue}{35},\textcolor{red}{15})  					&(\textcolor{blue}{26},\textcolor{red}{15}) 	&(\textcolor{blue}{18},\textcolor{red}{15})		\\ \hline
$\Lambda_4(X)$ -- \cite{BBC+95}			&(\textcolor{blue}{56},\textcolor{red}{56})		&(\textcolor{blue}{96},\textcolor{red}{28})					&(\textcolor{blue}{68},\textcolor{red}{38}) 	&(\textcolor{blue}{48},\textcolor{red}{28})		\\ \hline
\hspace{3.05em} -- \cite{NC04}				&(\textcolor{blue}{35},\textcolor{red}{35})		&(\textcolor{blue}{63},\textcolor{red}{23})					&(\textcolor{blue}{42},\textcolor{red}{23}) 	&(\textcolor{blue}{30},\textcolor{red}{23})		\\ \hline
$\Lambda_5(X)$ -- \cite{BBC+95}			&(\textcolor{blue}{84},\textcolor{red}{84})		&(\textcolor{blue}{134},\textcolor{red}{40})					&(\textcolor{blue}{104},\textcolor{red}{54}) 	&(\textcolor{blue}{72},\textcolor{red}{40})		\\ \hline
\hspace{3.05em} -- \cite{NC04}				&(\textcolor{blue}{49},\textcolor{red}{49})		&(\textcolor{blue}{97},\textcolor{red}{31})					&(\textcolor{blue}{60},\textcolor{red}{31}) 	&(\textcolor{blue}{42},\textcolor{red}{31})		\\ \hline
$\Lambda_{10}(X)$ -- \cite{BBC+95}			&(\textcolor{blue}{224},\textcolor{red}{224})	&(\textcolor{blue}{332},\textcolor{red}{100})				&(\textcolor{blue}{284},\textcolor{red}{134}) 	&(\textcolor{blue}{192},\textcolor{red}{100})	\\ \hline
\hspace{3.5em} -- \cite{NC04}				&(\textcolor{blue}{119},\textcolor{red}{119})	&(\textcolor{blue}{236},\textcolor{red}{71})					&(\textcolor{blue}{150},\textcolor{red}{71}) 	&(\textcolor{blue}{102},\textcolor{red}{71})		\\ \hline
GF($2^4$)-Mult \cite{MMC+09}				&(\textcolor{blue}{115},\textcolor{red}{112})	&(\textcolor{blue}{307},\textcolor{red}{68})					&(\textcolor{blue}{549},\textcolor{red}{56}) 	&(\textcolor{blue}{99},\textcolor{red}{68})		\\ \hline
GF($2^5$)-Mult \cite{MMC+09}				&(\textcolor{blue}{179},\textcolor{red}{175})	&(\textcolor{blue}{502},\textcolor{red}{115})				&(\textcolor{blue}{1250},\textcolor{red}{85})\footref{better} 	&(\textcolor{blue}{154},\textcolor{red}{115})	\\ \hline
GF($2^6$)-Mult \cite{MMC+09}				&(\textcolor{blue}{257},\textcolor{red}{252})	&(\textcolor{blue}{660},\textcolor{red}{150})				&(\textcolor{blue}{2449},\textcolor{red}{116})\footref{better}	&(\textcolor{blue}{221},\textcolor{red}{150})	\\ \hline
GF($2^7$)-Mult \cite{MMC+09}				&(\textcolor{blue}{349},\textcolor{red}{343})	&(\textcolor{blue}{996},\textcolor{red}{217})				&(\textcolor{blue}{2976},\textcolor{red}{179})\footref{better} 	&(\textcolor{blue}{300},\textcolor{red}{217})	\\ \hline
GF($2^8$)-Mult \cite{MMC+09}				&(\textcolor{blue}{469},\textcolor{red}{448})	&(\textcolor{blue}{1254},\textcolor{red}{264})				&(\textcolor{blue}{4441},\textcolor{red}{214})\footref{better} 	&(\textcolor{blue}{405},\textcolor{red}{264})	\\ \hline
GF($2^9$)-Mult \cite{MMC+09} 			&(\textcolor{blue}{575},\textcolor{red}{567})	&(\textcolor{blue}{1712},\textcolor{red}{351})				&(\textcolor{blue}{6552},\textcolor{red}{277})\footref{better} 	&(\textcolor{blue}{494},\textcolor{red}{351})	\\ \hline
GF($2^{10}$)-Mult \cite{MMC+09}			&(\textcolor{blue}{709},\textcolor{red}{700})	&(\textcolor{blue}{2206},\textcolor{red}{410})				&(\textcolor{blue}{7087},\textcolor{red}{351})\footref{better} 	&(\textcolor{blue}{609},\textcolor{red}{410})	\\ \hline
GF($2^{16}$)-Mult \cite{MMC+09}			&(\textcolor{blue}{1837},\textcolor{red}{1792})	&(\textcolor{blue}{6724},\textcolor{red}{1040})				&(\textcolor{blue}{*},\textcolor{red}{*})\footnote{\label{nogo}TOpt's TODD implementation cannot finish the optimization within 24 hours on a cluster node with Intel(R) Xeon(R) Platinum 8163 CPU @ 2.50GHz and 16GB of memory.}	 	&(\textcolor{blue}{1581},\textcolor{red}{1040})	\\ \hline
GF($2^{32}$)-Mult \cite{MMC+09}			&(\textcolor{blue}{7292},\textcolor{red}{7168})	&(\textcolor{blue}{33269},\textcolor{red}{4128})				&(\textcolor{blue}{*},\textcolor{red}{*})\footref{nogo} 	&(\textcolor{blue}{6268},\textcolor{red}{4128})	\\ \hline
GF($2^{64}$)-Mult \cite{MMC+09}			&(\textcolor{blue}{28860},\textcolor{red}{28672})&(\textcolor{blue}{180892},\textcolor{red}{16448})			&(\textcolor{blue}{*},\textcolor{red}{*})\footref{nogo} 	&(\textcolor{blue}{24765},\textcolor{red}{16448}) 	\\ \hline\hline
Average Reduction						&									&(\textcolor{blue}{+142.30\%},\textcolor{red}{-41.41\%})			&(\textcolor{blue}{+260.32\%},\textcolor{red}{-38.00\%}) 	&(\textcolor{blue}{0\%}\footnote{The difference between our CNOT-count and the original CNOT-count reported in Column 2 is caused by different decompositions of Toffoli gates. In our optimization procedure, we didn't further reduce CNOT-count though it looks we did from the numbers here. },\textcolor{red}{-42.67\%}) 	\\ \hline
Maximum T-count Reduction				&									&(\textcolor{blue}{+42.50\%}, \textcolor{red}{-65.71\%})			&(\textcolor{blue}{+402.04\%},\textcolor{red}{-61.22\%}) 	&(\textcolor{blue}{0\%}, \textcolor{red}{-71.43\%}) 	\\ \hline
 \end{tabular}
\caption{We report the CNOT-count and T-count after no optimization (original) and after T-par, TOpt and T-Optimizer optimizations with no ancillae.}
\label{tab:benchmarks2}
\end{table}

\section{Summary and Future Work}
In this work, we present a new optimization technique to reduce the number of \verb+T+ gates in Clifford+\verb+T+ circuits by treating every \verb+T+ gate conjugated by Clifford operators as $\frac{\pi}{4}$-rotations around Pauli operators. For benchmarking circuits like  $\textrm{Adder}_8$ and $\textrm{Mod}5_4$, T-Optimizer will reduce significantly more \verb+T+ gates than any other circuit optimization software. 

As we learn through the benchmarking result, all these benchmarked quantum circuit optimizers have ``sweet spots'' in which their performance has no equal. A unified framework of circuit optimization would be very interesting.  

Another avenue of work is to start from the original unitary gate \verb+U+ which may not be exactly represented by Clifford+\verb+T+. Many attempts have been made to address this fundamental open question.  An algorithm based on exhausted search was presented in \cite{F2011}. Given its exponential runtime, it is hard to practically use that for reasonable small accuracy $\epsilon$. For quantum algorithms based on phase estimations, phase kickback tricks will introduce ancillary qubits to perform phase gates\cite{KSV+02}. The Solovay-Kitaev algorithm and its variants are also well-studied in literature\cite{DN+05}. The algorithm runs in $O(\log^{2.71}(1/\epsilon))$ time and produce a quantum circuit of size $O(\log^{3.97}(1/\epsilon))$ to approximate the desired unitary gate up to accuracy $\epsilon$. In the past several years, several efficient algorithms with improved circuit size (compared with the Solovay-Kitaev algorithm) have been proposed for single-qubit unitary approximation\cite{KMM13, KMM+16, S12, RS14}. For instance, \cite{S12} presents an efficient algorithm to achieve \verb+T+-counts of $\sim 10+4\log_2(1/\epsilon)$, which matches the information-theoretic lower bound of $K+3\log_2(1/\epsilon)$. However, to efficiently approximate general multi-qubit unitary gates, there is still ample room for improvement in circuit size. 

\emph{Note added:} Simultaneously with our results,  Kissinger and Wetering demonstrated a method to reduce the number of \verb+T+-gates in a quantum circuit by presenting the quantum circuit as a ZX-diagram and then using the phase teleportation technique\cite{KW19}. Surprisingly, KW method produces benchmarking results identical to ours and both methods do not change number or locations of any non-phase gates.

%-----------------------------------------------------------------------------%
\section*{Acknowledgements}
The authors would like to thank Mario Szegedy and Neil Julien Ross for helpful discussions. J. C. would like to thank Yunseong Nam for sharing benchmarking circuit files used in \cite{CMN+18}. Their benchmarking circuits may involve some rotations which may not be Clifford gates anymore,  therefore the benchmark result is not included here.  
%-----------------------------------------------------------------------------%

\bibliographystyle{apsrev4-1}
\bibliography{Tdepth}

\appendix

%% Start-Of-Trailer

\end{document}